\documentclass[superscriptaddress, twocolumn, amsmath, amssymb, aps,
pra, letterpaper, notitlepage]{revtex4-1}

\usepackage{}

\usepackage{graphicx,graphics,epsfig,subfigure,times,bm,bbm,amssymb,amsmath,amsfonts,amsthm,mathrsfs,MnSymbol}
\usepackage[matrix,frame,arrow]{xypic}
\usepackage[normalem]{ulem}
\usepackage{slashed}
\usepackage{color}
\usepackage[usenames,dvipsnames,svgnames,table]{xcolor}
\definecolor{darkblue}{rgb}{0.0,0.0,0.3}
\usepackage[colorlinks=true,
            linkcolor=red,
            urlcolor= darkblue,
            citecolor=blue]{hyperref}

\usepackage{enumerate}
\usepackage{epstopdf}
\usepackage{relsize}

\newtheorem{definitionenv}{Definition}
\newtheorem{remarkenv}[definitionenv]{Remark}

\newtheorem{exampleenv}{Example}

\newtheorem{mydef}{Definition}

\newtheorem{mytheorem}{Theorem}
\newtheorem{mylemma}{Lemma}

\newcommand{\bes} {\begin{subequations}}
\newcommand{\ees} {\end{subequations}}
\newcommand{\bea} {\begin{eqnarray}}
\newcommand{\eea} {\end{eqnarray}}

\newcommand{\beq}{\begin{equation}}
\newcommand{\beqs}{\begin{equation*}}

\newcommand{\eeq}{\end{equation}}
\newcommand{\eeqs}{\end{equation*}}

\newcommand{\ignore}[1]{}

\def\>{\rangle}
\def\<{\langle}

\newcommand{\ket}[1]{|#1\rangle}

\newcommand{\cG}{\mathcal{G}}
\begin{document}
\title{Depth reduction for quantum Clifford circuits through  Pauli measurements}
%\title{Efficient Large Block Codes Ancilla States Preparation for Fault-tolerant Quantum Computation}
\author{Yi-Cong Zheng}
%\thanks{
%        Yicong Zheng and Todd A. Brun
%        are with the
%        Communication Sciences Institute, Electrical Engineering Department,
%University of Southern California, Los Angeles, California, USA  90089.
%        Mark M. Wilde is a postdoctoral fellow with the School of Computer
%Science, McGill University, Montreal, Quebec, Canada H3A 2A7.

%\affiliation{Department of Electrical Engineering, University of Southern California.}
\email{
zheng.yicong@quantumlah.org}
\affiliation{Centre for Quantum Technologies, National University of Singapore, Singapore 117543}
\affiliation
{Yale-NUS College, Singapore 138527}

\author{Ching-Yi Lai}
%\email{cylai0616@gmail.com}
\affiliation{Institute of Information Science, Academia Sinica, Taipei 11529, Taiwan}
\author{Todd A. Brun}
%\email{tbrun@usc.edu}
\affiliation{Ming Hsieh Department of Electrical Engineering, Center for Quantum Information Science and Technology, University of Southern California, Los Angeles, California 90089, USA\\}
\author{Leong-Chuan Kwek}
\affiliation{Centre for Quantum Technologies, National University of Singapore, Singapore 117543}
\affiliation{MajuLab, CNRS-UNS-NUS-NTU International Joint Research Unit, UMI 3654, Singapore}
\affiliation{
Institute of Advanced Studies, Nanyang Technological University, Singapore 639673}
\affiliation{National Institute of Education, Nanyang Technological University, Singapore 637616 }
\date{\today}

\begin{abstract}
%Stabilizer circuits play an important role in quantum computation. Aaronson and Gottesman have shown that any $n$ qubits stabilizer circuit on can be decomposed into $O(n^2/\log n)$ Hadamard, Phase and CNOT gates and the depth of circuit is $O(n)$. In this paper, we propose protocol to reduce the circuit depth to $O(1)$ by Pauli measurements. It needs additional $n$ auxiliary qubits and $O(1)$ different $4n-$qubit ancilla states for the purpose of Pauli measurements. These process can be shown to implement the same stabilizer circuit up to a permutation of qubits with the circuit depth $O(1)$. All necessary ancilla states can be identified to be two-colorable graph states, which can all be prepared fault-tolerantly and efficiently when encoded to quantum error-correcting codes.
%As such, it can greatly saves the \emph{in-situ} computation time to compute stabilizer circuits by off-line preparing a small set of stabilizer states.
%This result clearly shows that certain two-colorable graph states can be regarded as resources to speed up the quantum computation process.
Clifford circuits play an important role in quantum computation. Gottesman and Chuang proposed a gate teleportation protocol so
%that an $n$-qubit circuit $U$ of certain Clifford hierarchy can be implemented using the teleportation circuit.
that a quantum circuit can be implemented by the teleportation circuit with specific ancillary qubits.
%This is done by replacing the EPR pairs $|\Psi\>^{\otimes n}$ for teleportation with  $(I\otimes U)|\Psi\>^{\otimes n}$ and replacing  the Pauli correction with a correction of a lower-level Clifford hierarchy.
In particular, an $n$-qubit Clifford circuit $U$ can be implemented by preparing an ancillary stabilizer state  $(I\otimes U)|\Phi^+\>^{\otimes n}$  for teleportation and doing a Pauli correction
conditioned on the measurement.
%\rmark{The ancillary state can be (fault-tolerantly) prepared by $O(n)$ stabilizer (Pauli) measurements.}
In this paper, we provide an alternative procedure to implement a Clifford circuit through Pauli measurements, by preparing $O(1)$ ancillas that are Calderbank-Shor-Steane (CSS) stabilizer states.
That is to say, $O(1)$ CSS states are sufficient to implement any Clifford circuit. As an application to fault-tolerant quantum computation, any Clifford circuit can be implemented by $O(1)$ steps of Steane syndrome extraction
if clean CSS stabilizer states are available.

\end{abstract}
%\pacs{03.67.Lx, 03.67.Pp}
%\date{\today}
\maketitle

%\tableofcontents

\section{Introduction}
In quantum information and computation, the class of \emph{stabilizer circuits} can be efficiently simulated by classical computers~\cite{GK98} using the \emph{stabilizer formalism}~\cite{Gottesman:9705052}.
Stabilizer circuits are composed solely of Hadamard ({\rm H}), Phase ({\rm P}), and controlled-NOT ({\rm CNOT}) gates, defined as
\begin{equation*}
\text{H} = \frac{1}{\sqrt{2}}\left(
             \begin{array}{cc}
               1 & 1 \\
               1 & -1 \\
             \end{array}
           \right),\
\text{P} = \left(
             \begin{array}{cc}
               1 & 0 \\
               0 & i \\
             \end{array}
           \right),\
{\small \text{CNOT} = \left(
                \begin{array}{cccc}
                  1 & 0 & 0 & 0 \\
                  0 & 1 & 0 & 0 \\
                  0 & 0 & 0 & 1 \\
                  0 & 0 & 1 & 0 \\
                \end{array}
              \right),}
\end{equation*}
and  single-qubit Pauli measurements. A stabilizer circuit is called a \emph{Clifford circuit} if it contains no measurements,
and Hadamard, Phase, and CNOT gates are called \emph{Clifford gates}. It is known that the Clifford gates, together with any non-Clifford gate, form a universal set for quantum computation~\cite{Shi03}. Stabilizer circuits are especially important in fault-tolerant quantum computation (FTQC) for encoding, decoding, and error correction circuits~\cite{Nielsen:2000:CambridgeUniversityPress,Shor:1996:56,Aharonov:1997:176,Folwer2012PhysRevA.86.032324}, along with other applications, such as
evaluation of the average gate fidelity via randomized benchmarking~\cite{knill2008randomized,magesan2011scalable}, and efficient quantum simulations~\cite{Lloyd:1996:1073,aspuru2005simulated,wecker2014gate,Hastings:2015_simulation,Poulin:2015qic_simulation}.

It is well known that any $n$-qubit (unitary) quantum circuit $U$ of a certain level of the Clifford hierarchy can be implemented by \emph{gate teleportation}~\cite{Gottesman:1999:390}, which requires a $2n$-qubit ancilla state
\begin{align}
\left| \Phi_U^n \right\>=I\otimes U\left(\frac{|00\>+|11\>}{\sqrt{2}}\right)^{\otimes n}, \label{eq:UEPR}
\end{align}
and is performed by a \emph{single} step of Bell basis measurements followed by a controlled-correction circuit at a lower level of the Clifford hierarchy.
For a Clifford circuit $U$, $\left| \Phi_U^n \right\>$ is a \emph{stabilizer state}, which is a joint-$(+1)$ eigenvector of $2n$ commuting Pauli operators, called \emph{stabilizer generators}. The controlled-correction is simply a Pauli operator.
Thus, the complexity of a Clifford circuit is dominated by the preparation of $\left| \Phi_U^n \right\>$,
which can be prepared by measuring the stabilizer generators on $n$ EPR pairs (up to a Pauli correction).

At first sight, it seems as difficult to prepare  such an ancilla state   as to directly implement the circuit. However, in the case of FTQC, it is possibly easier to prepare specific known states for gate teleportation than to do gate operations on unknown states. One important  example is the \emph{magic state distillation} for the fault-tolerant implementation of non-Clifford gates~\cite{Bravyi:2005:022316,Bravyi_Haah_PhysRevA.86.052329}.
In some cases,  it may even be impossible to do gate operations directly on the qubits.  For example, for a FTQC scheme using multi-qubit quantum error-correcting codes~\cite{steane1999efficient_Nature,Steane:2003:042322,brun2015teleportation},  typically,  its fault-tolerant logical Clifford gates, if they exist, are computationally difficult to find.
Therefore we would like to investigate the implementation of stabilizer circuits by variants of gate teleportation in FTQC.

Consider an $n$-qubit Clifford circuit $U$. Previously, Gottesman and Chuang showed that the ancillary state $\ket{\Phi_U^n}$
can be fault-tolerantly prepared by a sequence of $O(n)$ fault-tolerant Pauli operator measurements, with error correction and verification inserted between each two consecutive measurements~\cite{Gottesman:1999:390}. This preparation is \emph{passive}  in that most of the procedure is error detection.
We will show that to implement a Clifford circuit, it suffices to do $O(1)$  gate teleportations with (clean) ancillas that are \emph{Calderbank-Shor-Steane (CSS) stabilizer states} (up to single-qubit Clifford gate operations), and hence can be fault-tolerantly prepared~\cite{Ancilla_distillation_1,zheng2017efficient}.
(A CSS state is defined by a set of stabilizer generators, each of which  can be  chosen to be the tensor product of identity and either $X$
or $Z$ Pauli operators.)
These ancilla states are thus equivalent to two-colorable graph states~\cite{chen2004multi}.

Our idea is motivated by Clifford circuit synthesis~\cite{aaronson2004improved,maslov2017Bruhat}. Aaronson and Gottesman showed that any Clifford circuit is equivalent to a circuit that contains 11 stages of computation in the sequence -H-C-P-C-P-C-H-P-C-P-C-~\cite{aaronson2004improved}, where {-H-,} {-P-,} and -C- stand for stages composed of only  Hadamard, Phase, and CNOT gates, respectively~\footnote{Consequently, any stabilizer circuit can be decomposed into $O(n^2/\log n)$ Clifford gates~\cite{markov2008optimal,aaronson2004improved} with circuit depth $O(n)$~\cite{kutin2007computation} or $O(n^2)$ Clifford gates with circuit depth $O(\log n)$~\cite{MN01}.}.
Recently, Maslov and  Roetteler  found that a Clifford circuit can be decomposed as a 9-stage sequence -C-P-C-P-H-P-C-P-C-~\cite{maslov2017Bruhat}.
Therefore, it suffices to implement each of the {-H-,} {-P-,} and -C- stages of the 11-stage or 9-state sequence for a Clifford circuit.
In FTQC, it is straightforward to combine Knill syndrome extraction~\cite{Knill2005} with gate teleportation~\cite{Gottesman:1999:390},
and clearly $\ket{\Phi_U^n}$, where $U$ is a {-H-,} {-P-,} or {-C-} circuit, is a CSS state up to single-qubit Clifford operations. Consequently, we can prepare the ancilla states for the 9-stage or 11-stage sequence by distillation~\cite{Ancilla_distillation_1,zheng2017efficient}.

On the other hand, it is not so obvious how to combine Steane syndrome extraction~\cite{steane1997active} with  gate teleportation.
Since both Steane and Knill syndrome extraction have their own advantages, we would like also to derive a \emph{constant-depth} gate teleportation procedure for Steane syndrome extraction.
For example, we remark that measurements of logical Pauli operators can be implemented simultaneously with error correction in Steane syndrome extraction~\cite{steane1997active},
and consequently we can have stabilizer circuits implemented solely with Steane syndrome extraction. Moreover, Steane syndrome measurements may lead to higher thresholds for certain CSS codes.
In this paper, we can propose such a procedure for Steane syndrome extraction through a series of Pauli measurements that implement the 9-stage sequence with the help of appropriate ancilla states that are CSS states up to single-qubit Clifford operations.
Again, these states
can be fault-tolerantly prepared.
We will discuss the procedure at the logical level: the underlying quantum error-correcting codes can be either single-qubit codes or multiple-qubit codes.
If the underlying FTQC scheme is based on multi-qubit quantum error-correcting codes, we have a roughly constant resource overhead~\cite{Ancilla_distillation_1,zheng2017efficient}.

  %Hence, as an application in fault-tolerant quantum computation (FTQC), when encoding into error-correcting codes, one can use Steane syndrome extraction to implement Clifford circuit in $O(1)$ steps.

% by introducing extra $n$ auxiliary qubits initially prepared in $|0\>^{\otimes n}$ or $|+\>^{\otimes n}$. The measurements themselves require to prepare $O(1)$ $4n$-qubit ancilla states in worst cases.

% Hence, such two-colorable graph states can be regarded as resources to speedup the stabilizer circuits.

%it needs $O(1)$ measurements of Pauli operators (Pauli measurements) to implement any stabilizer circuit, up to a permutation of qubits. Thus, the depth of the stabilizer circuit can be reduced to $O(1)$. The measurements themselves require to prepare a $4n-$ ancilla states. Consequently, one only needs to prepare $O(1)$ ancilla states in advance.  Thus, one can greatly reduce the \emph{in-situ} computation time for stabilizer circuits by off-line preparing a constant number of stabilizer states. Hence, the stabilizer states can be regarded as resources to speedup the stabilizer circuits.

The paper is organized as follows. We review preliminary material in Sec.~\ref{sec:prim}, including the stabilizer formalism and the representation of Clifford circuits.
%We also list all ancilla states required for FTQC.
In Sec.~\ref{sec:measurement}, we propose a method to measure arbitrary Pauli operators, and give conditions when several Pauli operators can be measured simultaneously.
In
Sec.~\ref{sec:decomposition}, we explicitly show how to perform  Clifford circuits via constant steps of Pauli measurements. Conclusions and discussion of the method are presented in Sec.~\ref{sec:discussion}.

% In Sec.~\ref{sec:logical ancilla}, we give the protocols for those entangled ancillas (with perfect distillation circuits).
%so that all these stabilizer ancillas states can be distilled if the distillation circuit is noiseless.
%In Sec.~\ref{sec:FTpreparation}, we propose a \emph{fully} fault-tolerant protocol for  these stabilizer ancilla states.
%In Sec.~\ref{sec:numerical}, we numerically study the procedure for preparing $|0\>_L$ for the quantum Golay code by estimating the error weight distribution of the output ancillas for different combinations of small-sized classical codes.
%Discussion and suggestions for further improvements are presented in Sec.~\ref{sec:discussion}.

\section{Preliminaries}\label{sec:prim}
%\rmark{why capital $N$?}
\subsection{Stabilizer formalism}
The Hilbert space of a single qubit is the two-dimensional complex vector space $\mathbb{C}^2$ with an orthonormal basis $\{|0\>, |1\>\}$. The Hilbert space of $N$-qubit states is hence $\mathbb{C}^{2^N}$.
Let $\mathcal{P}_N=\mathcal{P}_1^{\otimes N}$ denote the $N$-fold Pauli group, where
\begin{equation*}
\mathcal{P}_1=\{\pm I, \pm i I, \pm X, \pm i X, \pm Y, \pm i Y, \pm Z, \pm i Z\},
\end{equation*}
and $X={\footnotesize \left(
         \begin{array}{cc}
           0 & 1 \\
           1 & 0 \\
         \end{array}
       \right)}
$, $Z={\footnotesize \left(
         \begin{array}{cc}
           1 & 0 \\
           0 & -1 \\
         \end{array}
       \right)}$, and $Y=iXZ$ are the Pauli matrices.
%Here, we use the notation $X_j$ to denote an $X$ on qubit $j$, $I^{\otimes j-1}\otimes X\otimes I^{\otimes N-j}$ for short, where $N$ is the number of qubits. $Y_j$ and $Z_j$ are defined similarly.
Let $X_j$, $Y_j$, and $Z_j$ act as single-qubit Pauli matrices on the $j$th qubit and trivially elsewhere.
We also introduce the notation $X^{\mathbf a}$, for ${\mathbf a}=a_1\cdots a_N\in \mathbb{Z}_2^N$, to denote the operator $\otimes_{j=1}^N X^{a_j}$ and let $\text{supp}({\mathbf a})=\{j:a_j=1\}$.
For ${\mathbf a}, {\mathbf b}\in \mathbb{Z}_2^N$, denote the intersection of $\text{supp}({\bf a})$ and $\text{supp}({\bf b})$ by $\mathcal{I}_{\bf ab}$ and let $\tau_{\bf ab}=\left|\mathcal{I}_{\bf ab}\right|$.  An $N$-fold Pauli operator can be expressed as
\begin{equation}\label{eq:general_error}
i^l\cdot \bigotimes_{j=1}^N X^{a_j}Z^{b_j}=i^l X^{\bf a}Z^{\bf b}, \quad {\bf a},{\bf b}\in \mathbb{Z}^N_2, \ l\in\{0,1,2,3\}.
\end{equation}
Then $({\bf a}\,|\,{\bf b})$ is called the
\emph{binary representation} of the Pauli operator $i^lX^aZ^b$ up to an overall phase $i^l$. In particular, $\pm i^{\tau_{\bf ab}} X^{\bf a}Z^{\bf b}$ has eigenvalues $\pm 1$. From now on we use the binary representation, and we may neglect the overall phase for simplicity when there is no ambiguity.

For two Pauli operators $({\bf a}\, |\,  {\bf b})$ and $({\bf e}\, |\,  {\bf f})$, one can define their symplectic inner product:
\beqs
({\bf a}\, |\,  {\bf b}) J_N ({\bf e}\, |\,  {\bf f})^t=
\begin{cases}
&0, \quad \left[X^{\bf a}Z^{\bf b}, X^{\bf e}Z^{\bf f}\right]=0,\\
&1, \quad \left\{X^{\bf a}Z^{\bf b}, X^{\bf e}Z^{\bf f} \right\}=0,
\end{cases}
\eeqs
where
\begin{equation*}
 J_N =
\left(
\begin{array}{cc}
{\bf 0}_N & I_N \\
I_N & {\bf 0}_N \\
\end{array}
\right),
\end{equation*}
and $I_N$ and ${\bf 0}_N$ are the identity and zero matrices of dimension $N$, respectively. Here, $M^t$ denotes the transpose of $M$.
%Without ambiguity, we also use $a$ and $b$ to represent the index set $\{a_1,\dots a_n\}$ and $\{b_1,\dots b_n\}$.

Consider a set of commuting Pauli operators $\{G_1,\dots, G_{s}\}$ that does not generate $-I^{\otimes N}$. These Pauli operators  generate an Abelian subgroup (stabilizer group) $\mathcal{G}$ of $\mathcal{P}_N$, and thus are called the stabilizer generators of $\mathcal{G}$.
%Every element in $\mcal{P}_N$ has eigenvalues $\pm 1$.
Let $\mathcal{S}(\mathcal{G})$ denote the $2^{N-s}$-dimensional subspace of the $N$-qubit state space $\mathbb{C}^{2^N}$ fixed by $\mathcal{G}$,
which is the joint-$(+1)$ eigenspace of $G_1, \dots, G_{s}$.
Then for any $\ket{\psi}\in \mathcal{S}(\mathcal{G})$, one has $$G\ket{\psi}=\ket{\psi},$$ for all $G\in \mathcal{G}$. (Note that the overall phase of any $G\in\mathcal{G}$ can be $\pm 1$ only.)

A set of $s$ commuting $N$-fold Pauli operators has a binary representation as a matrix of the form:
\begin{equation*}
(A|B)= \left(
         \begin{array}{c|c}
           {\bf a}_1 & {\bf b}_1 \\
           \vdots & \vdots \\
           {\bf a}_s & {\bf b}_s \\
         \end{array}
       \right).
\end{equation*}
%The commutative relation requires $(A|B) J_N (A|B)^t = {\bf 0}_s$,

\begin{mydef}[Symplectic partner]\label{def:symplectic_partner}
For a set of $s$ commuting $N$-fold Pauli operators $(A|B)$, its symplectic partner $(E|F)$ is a  set of $s$ commuting  $N$-fold Pauli operators satisfying the orthogonality relation with respect to the symplectic inner product:
\beqs
(A|B)J_N(E|F)^t=I_s.
\eeqs
\end{mydef}
Note that if $(A|B)$ is of rank less than $N$, its symplectic partner is not unique.

\subsection{Clifford circuits}\label{sec:stabilizer_circuit}

An  $N$-qubit Clifford circuit can be represented by a $2N\times 2N$ binary matrix with respect to the basis of the binary representation of Pauli operators in (\ref{eq:general_error}).
For example, the idle circuit (no quantum gates) is represented by $I_{2N}$, the $2N\times 2N$  identity matrix.
The representation of consecutive Clifford circuits $M_1,\dots, M_j$  is their binary matrix product $ M_1 \cdots M_j$.
%, where $M_0=I_{2N}$ represents the idle circuit.

The $N$-qubit Clifford circuits form a finite group, which, up to overall phases, is isomorphic to the binary symplectic matrix group defined as follows:~\cite{aaronson2004improved}
\begin{mydef}[Symplectic group]\label{def:symplectic_group}
The group of $2N\times 2N$ symplectic matrix over $\mathbb{Z}_2$ is defined as:
\beqs
{\rm Sp}(2N, \mathbb{Z}_2)\equiv \{M\in {\rm GL}(2N, \mathbb{Z}_2): MJ_NM^t=J_N\}
\eeqs
under matrix multiplication.
\end{mydef}
In general, $M\in {\rm Sp}(2N, \mathbb{Z}_2)$ has the form
\beq
M=\left(
  \begin{array}{c|c}
    Q & R \\
    \hline
    \rule[0.4ex]{0pt}{8pt}
    S & T \\
  \end{array}
\right),  \label{eq:symplectic_matrix}
\eeq
where $Q$, $R$, $S$ and $T$ are $N\times N$ square matrices satisfying the following conditions:
\beq
QR^t=RQ^t, \quad ST^t = TS^t, \quad Q^tT + R^t S = I_N.
\eeq
In other words, $(Q|R)$ is a symplectic partner of $(S|T)$ by Def.~\ref{def:symplectic_partner}. Unlike  Ref.~\cite{aaronson2004improved}, here we omit the column vector that corresponds to the phases ($\pm 1$ only) of the operators. If needed, such overall phases can always be compensated by a single layer of gates consisting solely of $Z$ and $X$ gates~\footnote{Such extra layer has depth $O(1)$. Throughout the paper, Pauli gates are assumed to be free and can be directly applied to qubits. This is also true in FTQC using stabilizer codes, where logical Pauli operators are easy to realize. } on some subsets of qubits~\cite{aaronson2004improved,maslov2017Bruhat}.

Let $\text{C}(j,l)$  denote a CNOT gate with control qubit $j$ and target qubit $l$.
The actions of appending a Hadamard, Phase, or CNOT gate to  a Clifford circuit $M$ can be described as follows:
\begin{enumerate}
  \item A Hadamard gate on qubit $j$ exchanges columns $j$ and $N+j$ of $M$.
  \item A Phase gate on qubit $j$ adds column $j$ to column $N+j$ (modulo 2) of $M$.
  \item $\text{C}(j,l)$ adds column $j$ to column $l$ (modulo 2) of $M$ and adds column $N+l$ to column $N+j$ (modulo 2) of $M$.
\end{enumerate}

Now, consider a $2^k$ dimensional subspace $\mathcal{S}(\mathcal{G})$ of the $N$-qubit space, where $\cG$ has $k\leq N$ stabilizer generators.
$\mathcal{S}(\mathcal{G})$ encodes $k$ ``logical" qubits. We focus on the effects of Clifford circuits on these logical qubits in the stabilizer formalism. Consider a set of matrices $\textsf{C}_{\mathcal{G}}$ of the form:
\beq\label{eq:generalized_circuit matrix}
\left(
  \begin{array}{c|c}
    Q' & R' \\
    \hline
    \rule[0.4ex]{0pt}{8pt}
    S' & T'  \\
    \hline
    \rule[0.4ex]{0pt}{8pt}
    A & B
  \end{array}
\right).
\eeq
Here, $(A|B)$ corresponds to the stabilizer generators of $\mathcal{G}$;   $(Q'|R')$ and $(S'|T')$ are $k\times 2N$ binary matrices orthogonal to $(A|B)$ with respect to the symplectic inner product, and which are symplectic partners of each other. They can  be regarded as ``logical operators" on $\mathcal{S}(\mathcal{G})$.
 %\bmark{which encodes the information of stabilizer circuits acting on $\mathcal{S}(\mathcal{G})$.}
%Denote the linear space spanned by vectors $(A|B)$ as $\mathscr{G}=\text{Span}\{(A|B)\}$.
We define the following equivalence relation $R$ in  $\textsf{C}_\mathcal{G}$: Two matrices
\beqs
{C}_1=\left(
  \begin{array}{c|c}
    Q'_1& R'_1 \\[2pt]
    \hline
    \rule[0.4ex]{0pt}{8pt}
    S'_1 & T'_1  \\[2pt]
    \hline
    \rule[0.4ex]{0pt}{8pt}
    A_1 & B_1
  \end{array}
\right) \quad \text{and} \quad {C}_2=\left(
  \begin{array}{c|c}
    Q'_2& R'_2 \\[2pt]
    \hline
    \rule[0.4ex]{0pt}{8pt}
    S'_2 & T'_2  \\[2pt]
    \hline
    \rule[0.4ex]{0pt}{8pt}
    A_2 & B_2
  \end{array}
\right),
\eeqs
are equivalent if (a)~$(A_1|B_1)$ and $(A_2|B_2)$ generate the same stabilizer group $\mathcal{G}$; and (b)~{\tiny
$\left(
  \begin{array}{c|c}
    Q'_1 & R'_1 \\[2pt]
    \hline
    \rule[0.4ex]{0pt}{5pt}
    S'_1 & T'_2  \\
  \end{array}
\right)$} differs from {\tiny
$\left(
  \begin{array}{c|c}
    Q'_2 & R'_2 \\[2pt]
    \hline
    \rule[0.4ex]{0pt}{5pt}
    S'_2 & T'_2  \\
  \end{array}
\right)$} by multiplication of elements in $\mathcal{G}$. Thus, there is a one-to-one correspondence between $\textsf{C}_{\mathcal{G}}/R$ and ${\rm Sp}(2k, \mathbb{Z}_2)$.
%Obviously, when $k=N$, $(A|B)$ vanishes and $\mathcal{C}$ reduces to {\scriptsize
%$\left(
%  \begin{array}{c|c}
%    W' & X' \\
%    \hline
%    \rule[0.4ex]{0pt}{8pt}
%    Y' & Z'  \\
%  \end{array}
%\right)$}, which is a symplectic matrix.

Therefore, $\textsf{C}_{\mathcal{G}}/R$ captures the behavior of stabilizer circuits on $\mathcal{S}(\mathcal{G})$. The circuit representation of Eq.~(\ref{eq:generalized_circuit matrix}) is called the \emph{generalized stabilizer form} (GSF) of  a stabilizer subspace throughout the paper.
It will be used as the starting point of the discussion in the rest of the paper.

\section{Measurements of Pauli Operators}\label{sec:measurement}
\subsection{Measurement of an arbitrary single Pauli operator}
We consider the measurement of an arbitrary Hermitian Pauli operator $\pm i^{\tau_{\bf ab}}X^{\bf a}Z^{\bf b}$ on $N$ qubits,
where $\bf a$, ${\bf b}\in \mathbb{Z}_2^N$.
%$
%&|\Omega_{ab}\rangle\\
%=|0\rangle^{\otimes |a|}|+\rangle^{\otimes |b|}+|1\rangle^{\otimes |a|}|-\rangle^{\otimes |b|}
%$

%\subsubsection{$X^{a}$ and $Z^{b}$ operators}
%
%\subsubsection{$X^aZ^b$ operators}
\begin{figure}[!htp]
\centering\includegraphics[width=80mm]{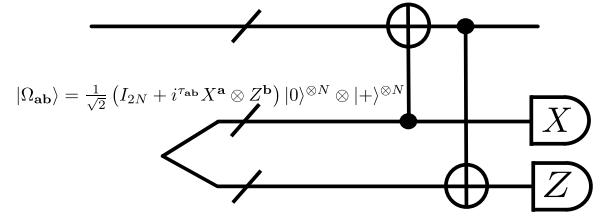}
\caption{\label{fig:measurement} The measurement circuit for $i^{\tau_{\bf ab}}X^{\bf a}Z^{\bf b}$, where $|\Omega_{\bf ab}\>$ is an ancilla state of two blocks of $N$ qubits. After two steps of transversal CNOT gates, the first and second blocks of ancillas are bitwise measured in the $X$ and $Z$ bases, respectively.
}
\end{figure}
The measurement of $i^{\tau_{\bf ab}} X^{\bf a}Z^{\bf b}$ can be realized by the circuit in Fig.~\ref{fig:measurement} with two blocks of ancilla qubits, each containing $N$ qubits. The $2N$-qubit ancilla is prepared in the special state:
\beq
|\Omega_{\bf ab}\> =\frac{1}{\sqrt{2}}\left(I_{2N}+i^{\tau_{\bf ab}}X^{\bf a}\otimes  Z^{\bf b}\right)|0\rangle^{\otimes N}\otimes |+\rangle^{\otimes N}.
\eeq
It is easy to see that it is a stabilizer state and thus, it can be prepared by a Clifford circuit.
%\rmark{(say something about how the measurement outcomes imply the Pauli measurement)}

Now we prove the functionality of the circuit in Fig.~\ref{fig:measurement}.
%when $\mathcal{I}_{\bf ab}=\emptyset$. The generalization to $\mathcal{I}_{\bf ab}\neq \emptyset$ is trivial.
We start with the joint state $|\psi\>|\Omega_{\bf ab}\>$. After two transversal CNOTs, the state becomes
\beq
\frac{1}{\sqrt{2}}\left(|\psi\> |0\>^{\otimes N}|+\>^{\otimes N} + i^{\tau_{\bf ab}} X^{\bf a}Z^{\bf b} |\psi\>  X^{\bf a} |0\>^{\otimes N} Z^{\bf b}|+\>^{\otimes N}\right).
\eeq
Let the measurement outcome of the $j$th qubit in the first and second blocks be $v^{x}_j$ and $v^{z}_j\in \{0,1\}$, respectively. Then the joint output state is:
\begin{widetext}
\beq
\begin{split}
&\frac{1}{\sqrt{2}}|\psi\>\bigotimes_{j=1}^N\left(\frac{I+(-1)^{v_j^x}X}{2}|0\>\right)\bigotimes_{j=1}^{N}\left(\frac{I+(-1)^{v_j^z}Z}{2}|+\>\right)+\frac{1}{\sqrt{2}}i^{\tau_{\bf ab}}X^{\bf a}Z^{\bf b}|\psi\>\bigotimes_{j=1}^{N}\left(\frac{I+(-1)^{v_j^x}X}{2}X^{a_j}|0\>\right)\bigotimes_{j=1}^{N}\left(\frac{I+(-1)^{v_j^z}Z}{2}Z^{b_j}|+\>\right)\\
=&\frac{1}{\sqrt{2}}\left(I+\prod_{l\in \text{supp}({\bf a})}(-1)^{v_l^x}\prod_{l\in \text{supp}({\bf b})}(-1)^{v_l^z}i^{\tau_{\bf ab}}X^aZ^b\right)|\psi\>\bigotimes_{j=1}^{N}\left(\frac{I+(-1)^{v_j^x}X}{2}|0\>\right)\bigotimes_{j=1}^{N}\left(\frac{I+(-1)^{v_j^z}Z}{2}|+\>\right),\\
\end{split}
\eeq
\end{widetext}
which is the state after the measurement of $i^{\tau_{\bf ab}}X^{\bf a}Z^{\bf b}$ on $|\psi\>$ with measurement outcome $\prod_{l\in \text{supp}({\bf a})}(-1)^{v_l^x}\prod_{l\in \text{supp}({\bf b})}(-1)^{v_l^z}$.
Thus the circuit works as we claimed.
This Pauli measurement is especially useful when one wants to measure several Pauli operators simultaneously, as we will see in the next subsection.

\subsection{Simultaneous measurement of multiple Pauli operators}\label{sec:multi_measurement}
One may wish to measure several Pauli operators simultaneously.
For a set of non-commuting operators this is not possible, since measuring these operator in different time orders may lead to different final states even with the same measurement outcomes.
However, if the set of Pauli operators commutes, this can be easily done by the circuit in Fig.~\ref{fig:measurement}. In this paper, we restrict ourselves to a commuting set of $d \leq N$ Pauli operators. Suppose the set of commuting Pauli operators to be measured is
\beqs
\{X^{{\bf e}_1}Z^{{\bf f}_1},\dots, X^{{\bf e}_d}Z^{{\bf f}_d} \}.
\eeqs
It is easy to see that one can measure these operators simultaneously by replacing $|\Omega_{\bf ab}\>$ in Fig.~\ref{fig:measurement} with the following stabilizer state:
\beq\label{eq:state_multi_measurement}
|\Omega_{{\bf EF}}\>=\frac{1}{\sqrt{2^d}}\prod_{j=1}^{d}
\left(I_{2N}+i^{\tau_{{\bf e}_j{\bf f}_j}}X^{{\bf e}_j}\otimes Z^{{\bf f}_j}\right)|0\>^{\otimes N}\otimes |+\>^{\otimes N}.
\eeq
$|\Omega_{{\bf EF}}\>$ is also a stabilizer state and can be prepared by a Clifford circuit.

It is also useful to check how GSF changes after a simultaneous measurement of multiple Pauli operators. Here, we consider the special case when $d=N-k$, the number of independent stabilizer generators of $\mathcal{G}$. Then one has the following statement from the theory of the stabilizer formalism:
\begin{mylemma}\label{lemma:multi_measurement}
Consider a circuit with GSF of the form Eq.~(\ref{eq:generalized_circuit matrix}) and a set of $N-k$ (independent) commuting Pauli operators:
\beqs
(E|F)=\left(
        \begin{array}{c|c}
         {\bf e}_1 & {\bf f}_1\\
         \vdots & \vdots\\
         {\bf e}_{N-k} & {\bf f}_{N-k} \\
        \end{array}
      \right).
\eeqs
If the following conditions are satisfied:
\begin{enumerate}
  \item $(E|F)J_N(E|F)^t={\bf 0}_{N-k}$;
  \item $(A|B)J_N (E|F)^t=I_{N-k}$;
  \item $(Q'|R')J_N(E|F)^t={\bf 0}$;
  \item $(S'|T')J_N(E|F)^t={\bf 0}$;
\end{enumerate}
then the GSF of the circuit after the simultaneous measurements of $(E|F)$ becomes
\beqs
\left(
  \begin{array}{c|c}
    Q' & R' \\
    \hline
    \rule[0.4ex]{0pt}{8pt}
    S' & T'  \\
    \hline
    \rule[0.4ex]{0pt}{8pt}
    E & F
  \end{array}
\right).
\eeqs
\end{mylemma}
The first two conditions state that $(E|F)$ is a symplectic partner of stabilizer generators $(A|B)$, while the third and forth imply that $(E|F)$ is orthogonal to  the logical operators. Note that the measurement outcomes are encoded in the overall phases and hence are not explicitly shown in this discussion.

%%\section{Stages via Pauli measurements}
%\subsection{Equivalence of two side measurement}
%(need a better name)

\section{Clifford circuits via  a constant number of measurement steps}\label{sec:decomposition}
In this section, we consider Clifford circuits consisting of $n$ qubits. We provide a constructive proof to show that by introducing $n$ extra auxiliary qubits, an arbitrary Clifford circuit can be implemented via a constant number of Pauli operator measurements, up to a permutation of qubits.

For clarity, we label the original $n$ data qubits as $\text{Q}_1,\dots, \text{Q}_n$, and the auxiliary qubits as $\text{A}_1, \dots, \text{A}_n$.
Now we have a total of $N=2n$ qubits in the order  $\{\text{A}_1,\dots, \text{A}_n, \text{Q}_1,\dots, \text{Q}_n\}$.
Suppose that we want to implement a Clifford circuit
  {\tiny $\left(
  \begin{array}{c|c}
    C_1 & C_2 \\[2pt]
    \hline
    \rule[0.4ex]{0pt}{5pt}
    C_3 & C_4  \\
  \end{array}
\right)$}  on $\text{Q}_1,\dots, \text{Q}_n$.
As in Sec.~\ref{sec:stabilizer_circuit},
the GSF with $n$ stabilizer generators (corresponding to the auxiliary qubits) can be written in the form of Eq.~(\ref{eq:generalized_circuit matrix}).
%\beqs
%\left(
%  \begin{array}{c|c}
%    W' & X' \\
%    \hline
%    \rule[0.4ex]{0pt}{8pt}
%    Y' & Z'  \\
%    \hline
%    \rule[0.4ex]{0pt}{8pt}
%    A & B
%  \end{array}
%\right).
%\eeqs
%where {\tiny $\left(
%  \begin{array}{c|c}
%    W & X \\[2pt]
%    \hline
%    \rule[0.4ex]{0pt}{8pt}
%    Y & Z  \\
%  \end{array}
%\right)$} is symplectic matrix.
 If the initial state of the auxiliary qubits is $|+\>^{\otimes n}$ or $|0\>^{\otimes n}$, then we start with the GSF of the idle circuit:
\beq\label{eq:identity}
\mathcal{I}=\left(
  \begin{array}{cc|cc}
    {\bf 0}_n  & I_n & {\bf 0}_n  & {\bf 0}_n \\
    \hline
    \rule[0.4ex]{0pt}{8pt}
    {\bf 0}_n  & {\bf 0}_n & {\bf 0}_n & I_n \\
    \hline
    \rule[0.4ex]{0pt}{8pt}
    I_n & {\bf 0}_n  & {\bf 0}_n & {\bf 0}_n \\
  \end{array}
\right) \quad \text{or} \quad
\left(
  \begin{array}{cc|cc}
    {\bf 0}_n  & I_n & {\bf 0}_n  & {\bf 0}_n \\
    \hline
    \rule[0.4ex]{0pt}{8pt}
    {\bf 0}_n  & {\bf 0}_n & {\bf 0}_n & I_n \\
    \hline
    \rule[0.4ex]{0pt}{8pt}
    {\bf 0}_n & {\bf 0}_n  & I_n & {\bf 0}_n \\
  \end{array}
\right),
\eeq
and end up with
\beq\label{eq:generalzied_circuit_2n}
\mathcal{C}=\left(
  \begin{array}{cc|cc}
    {\bf 0}_n  & C_1 & {\bf 0}_n  & C_2 \\
    \hline
    \rule[0.4ex]{0pt}{8pt}
    {\bf 0}_n  & C_3 & {\bf 0}_n & C_4 \\
    \hline
    \rule[0.4ex]{0pt}{8pt}
    I_n & {\bf 0}_n & {\bf 0}_n & {\bf 0}_n \\
  \end{array}
\right) \quad \text{or} \quad \left(
  \begin{array}{cc|cc}
    {\bf 0}_n  & C_1 & {\bf 0}_n  & C_2 \\
    \hline
    \rule[0.4ex]{0pt}{8pt}
    {\bf 0}_n  & C_3 & {\bf 0}_n & C_4 \\
    \hline
    \rule[0.4ex]{0pt}{8pt}
     {\bf 0}_n & {\bf 0}_n & I_n & {\bf 0}_n \\
  \end{array}
\right) .
\eeq
%whereone wishes to compute.
The goal is to find a sequence of Pauli measurements that transforms the initial circuits in Eq.~(\ref{eq:identity}) into the circuits in Eq.~(\ref{eq:generalzied_circuit_2n}).
Equivalently, one can start with Eq.~(\ref{eq:generalzied_circuit_2n}) and reduce the matrix to the initial circuits through Pauli measurements.
%This will also reveal all the real Pauli measurements one needs to apply.

%Similar to $R_1$ and $R_2$, $R_3$ is an $n\times 1$ matrix with entries $1$ and $-1$, which represent the signs of stabilizer generators. Such phases can always be fixed to 1 and be ignored from now on. Then one can define the canonical form of the stabilizer circuit:
%\beq
%\left(
%  \begin{array}{cc|cc}
%    {\bf 0}_n  & A & {\bf 0}_n  & C \\
%    \hline
%    {\bf 0}_n  & B & {\bf 0}_n & D \\
%    \hline
%    A & B & C & D  \\
%  \end{array}
%\right)
%\eeq

\subsection{9-stage Clifford circuit decomposition}
To find a sequence of Pauli measurements that implement a Clifford circuit,
the first step is to decompose the Clifford circuit into simple stages, each of which only contains a single type of Clifford gates. It is known that any Clifford circuit has an equivalent circuit that contains an 11-stage computation as -H-C-P-C-P-C-H-P-C-P-C-~\cite{aaronson2004improved}. Recently it was shown that a further reduction to a 9-stage computation -C-P-C-P-H-P-C-P-C- is possible~\cite{maslov2017Bruhat}.
We will consider the 9-stage decomposition in the following discussion. More specifically, one has the following  result:
\begin{mytheorem}[Bruhat decomposition~\cite{maslov2017Bruhat}]\label{thm:bruhat}
Any symplectic matrix $M$ of dimension $2n\times 2n$ can be decomposed as
%\beq\label{eq:bruhat_decomposition}
%M=T_1V_1T_2V_2 H_1 \pi V_3T_3V_4T_4.
%\eeq
\beq\label{eq:used_bruhat}
\begin{split}
M=&M^{(1)}_{C} M^{(1)}_{P}M^{(2)}_{C}M^{(2)}_{P} M^{(1)}_{H} \cdot\\
&M^{(3)}_{P} \left(\pi M^{(3)}_{C}\pi^{-1}\right)M^{(4)}_{P} \left(\pi M^{(4)}_{C}\pi^{-1}\right)\pi.
\end{split}
\eeq
Here, $M^{(j)}_{C} $ are -C- stage matrices containing only CNOT gates $\text{C}(q,r)$ such that $q < r$; $M^{(j)}_{P} $ and $M^{(j)}_{H} $ represent matrices of -P- and -H- stages; $\pi$ is a permutation matrix.
\end{mytheorem}

Compared to the 11-stage decomposition, the 9-stage decomposition has fewer stages, and it only requires CNOTs such that the index of the control qubit is less than the index of the target qubit.
Recall that a symplectic matrix can be expressed as in Eq.~(\ref{eq:symplectic_matrix}).
The corresponding  {symplectic} matrix of a -C- stage with such CNOTs can be written as
\beq
M_C=\left(
  \begin{array}{c|c}
    U      &  {\bf 0}_n\\[2pt]
    \hline
    \rule[1.0ex]{0pt}{8pt}
    {\bf 0}_n &  \left(U^{t}\right)^{-1}\\
  \end{array}
\right),
\eeq
where $U$ is an invertible $n\times n$ upper triangular matrix. As an example, a circuit of two consecutive CNOT gates is shown in Fig.~\ref{fig:cnots} and
its symplectic matrix   is
\beqs
\left(
  \begin{array}{ccc|ccc}
    1 & 1 & 0 & 0 & 0 & 0 \\
    0 & 1 & 1 & 0 & 0 & 0 \\
    0 & 0 & 1 & 0 & 0 & 0 \\
    \hline
    0 & 0 & 0 & 1 & 0 & 0 \\
    0 & 0 & 0 & 1 & 1 & 0 \\
    0 & 0 & 0 & 1 & 1 & 1 \\
  \end{array}
\right).
\eeqs
\begin{figure}[!htp]
\centering\includegraphics[width=50mm]{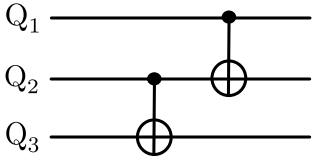}
\caption{\label{fig:cnots} Two consecutive CNOTs gates.
}
\end{figure}
There are approximately $n^2/2$ different such CNOT circuits~\cite{maslov2017Bruhat}. As we will see soon, this property is particularly useful when one tries to implement -C- stages via Pauli measurements.

For a -P- stage, since $\text{P}^4=I_2$, effectively   there are three single-qubit gates: P, P$^2=Z$ and $\text{P}^3=\text{P}^\dag=\text{P}Z$. Note that we will postpone all the $Z$ gates to the final stage, and thus the -P- layer consists of at most $n$ individual Phase gates. Hence, the symplectic matrix of a -P- stage is in general of the form:
\beq
M_P=\left(
  \begin{array}{c|c}
    I_n      &  \Lambda \\
    \hline
    \rule[0.4ex]{0pt}{8pt}
    {\bf 0}_n &  I_n\\
  \end{array}
\right),
\eeq
where $\Lambda$ is a diagonal matrix.

Similar to the -P- stage,  since $\text{H}^2=I_2$, an -H- stage  contains at most $n$ individual H gates. The symplectic matrix of an -H- stage on an arbitrary set of $m$ qubits can be written as $M_{H} = \pi' M_{H_m} \pi'^{-1}$, where
\beq
M_{H_m}=  \left(
        \begin{array}{cccc}
          {\bf 0} & {\bf 0} & I_m & {\bf 0} \\
          {\bf 0} & I_{n-m} & {\bf 0} & {\bf 0} \\
          I_m & {\bf 0} & {\bf 0} & {\bf 0} \\
          {\bf 0} & {\bf 0} & {\bf 0} & I_{n-m} \\
        \end{array}
      \right),
\eeq
represents the Hadamard gates acting on $\text{Q}_1,\dots, \text{Q}_m$ and $\pi'$ is some permutation matrix.

%Note that Eq.~(\ref{eq:bruhat_decomposition}) can also be written in another form:
%\beq\label{eq:used_bruhat}
%M=T_1V_1T_2V_2 H_1 V_3' \left(\pi T_3\pi^{-1}\right)V_4' \left(\pi T_4\pi^{-1}\right)\pi.
%\eeq
%where $V_j'=\pi V_j\pi^{-1}$, $j=3,4$ and the combination of $T$ and $\pi$ can form a single -C- stage. Thus, one restores the 9-stage decomposition -C-P-C-P-H-P-C-P-C-. In the rest of this section, we will study to compute each stage of the decomposition Eq.~(\ref{eq:used_bruhat}) via Pauli measurements.

\subsection{The -P- stage}
As discussed above, a -P- stage only contains at most a single Phase gate acting on each qubit, and thus we discuss the effect of a single-qubit Phase gate on data qubit $\text{Q}_j$.
For $m$ Phase gates acting on a subset of $m$ qubits, such a procedure can be done simultaneously by Pauli measurements.

Consider a pair of qubits $\{\text{A}_j,\text{Q}_j\}$ with $\text{A}_j$ in $|0\>$ state. The GSF of the idle circuit is:
\beqs
\left(
  \begin{array}{cc|cc}
    0 & 1 & 0 & 0 \\
        \hline
    \rule[0.4ex]{0pt}{8pt}
    0 & 0 & 0 & 1 \\
        \hline
    \rule[0.4ex]{0pt}{8pt}
    0 & 0 & 1 & 0 \\
  \end{array}
\right),
\eeqs
where the first and second columns correspond to $\text{A}_j$ and $\text{Q}_j$, respectively;
the first two rows are the logical operators corresponding to $\text{Q}_j$ and the third row represents the stabilizer generator corresponding to $\text{A}_j$.

First, add the stabilizer row to both logical operator rows (which will give an equivalent GSF of the circuit), and measure  operator $(1 \ \ 1 \ | \ 0 \ \ 1 )$ or $X_{\text{A}_j}Y_{\text{Q}_j}$. One obtains
\beqs
\left(
  \begin{array}{cc|cc}
    0 & 1 & 1 & 0 \\
        \hline
    \rule[0.4ex]{0pt}{8pt}
    0 & 0 & 1 & 1 \\
        \hline
    \rule[0.4ex]{0pt}{8pt}
    1 & 1 & 0 & 1 \\
  \end{array}
\right),
\eeqs
which is equivalent to
\beqs
\left(
  \begin{array}{cc|cc}
    1 & 0 & 1 & 1 \\
        \hline
    \rule[0.4ex]{0pt}{8pt}
    0 & 0 & 1 & 1 \\
        \hline
    \rule[0.4ex]{0pt}{8pt}
    1 & 1 & 0 & 1 \\
  \end{array}
\right)
\eeqs
by adding the stabilizer row to the first logical operator row. Next do the Pauli measurement  $(0 \ \ 0 \ | \ 0 \ \ 1 )$ or $Z_{\text{Q}_j}$. One gets
\beqs
\left(
  \begin{array}{cc|cc}
    1 & 0 & 1 & 0 \\
        \hline
    \rule[0.4ex]{0pt}{8pt}
    0 & 0 & 1 & 0 \\
        \hline
    \rule[0.4ex]{0pt}{8pt}
    0 & 0 & 0 & 1 \\
  \end{array}
\right).
\eeqs
After swapping $\text{A}_j$ and $\text{Q}_j$, the overall effect is a Phase gate on $\text{Q}_j$ up to a Pauli correction depending on the measurement outcomes. The swap does not need to be done physically. Instead, one can just keep a record of it in software.

For the case of $m$ Phase gates, since $\{X_{\text{A}_j}Y_{\text{Q}_j}| 1\leq j\leq n\}$ and $\{Z_{\text{Q}_j}\ | \ 1\leq j\leq n \}$ are commuting operator sets and the measurements of $\{Z_{\text{Q}_j}\}$ can be directly applied, it requires only two steps of Pauli measurement and one $4n$-qubit ancilla state for a -P- stage. If Phase gates are applied to
 a set $\mathscr{M}$ of qubits, then the required ancilla state is
{\small
\beq\label{eq:phase_ancilla}
|\Omega_{\text{P}_\mathscr{M}}\>=\frac{1}{\sqrt{2^{|\mathscr{M}|}}}\prod_{j\in \mathscr{M}}
\left(I_{2N}+i\left(X_jX_{j+n}\right)\otimes Z_{j+n}\right)|0\>^{\otimes 2n}\otimes |+\>^{\otimes 2n}.
\eeq
}This state can be obtained by projecting $|0\>^{\otimes 2n}\otimes |+\>^{\otimes 2n}$ to the joint $+1$ eigenspace of $\{X_jX_{j+n}\otimes Z_{j+n}\ |\ j\in \mathscr{M}\}$. Thus, it is stabilized by $\{Z_j Z_{j+n}\otimes I_{2n}, Z_{j+n}\otimes X_{j+n},X_jX_{j+n}\otimes Z_{j+n}\ |\ j\in \mathscr{M}\}$, which is a CSS state up to Hadamard gates on qubits $\{j+n, \ j\in \mathscr{M}\}$ in the second ancilla block.

\subsection{The -H-stage}
Like the -P- stage, we  consider only a single H on a data qubit. For a pair of qubits $\{\text{A}_j,\text{Q}_j\}$ with $\text{A}_j$ in $|0\>$ state, the idle circuit is
\beqs
\left(
  \begin{array}{cc|cc}
    0 & 1 & 0 & 0 \\
        \hline
    \rule[0.4ex]{0pt}{8pt}
    0 & 0 & 0 & 1 \\
        \hline
    \rule[0.4ex]{0pt}{8pt}
    0 & 0 & 1 & 0 \\
  \end{array}
\right).
\eeqs
Adding the stabilizer row to the first row of the logical operator and then measuring $(1\ \ 0 \ | \ 0 \ \ 1)$ or $X_{A_j}Z_{Q_j}$, one obtains
\beqs
\left(
  \begin{array}{cc|cc}
    0 & 1 & 1 & 0 \\
        \hline
    \rule[0.4ex]{0pt}{8pt}
    0 & 0 & 0 & 1 \\
        \hline
    \rule[0.4ex]{0pt}{8pt}
    1 & 0 & 0 & 1 \\
  \end{array}
\right).
\eeqs
Adding the stabilizer row to the second row of the logical operator and measuring $(0\ \ 1 \ | \ 0 \ \ 0)$ or $X_{\text{Q}_j}$, one gets
\beqs
\left(
  \begin{array}{cc|cc}
    0 & 1 & 1 & 0 \\
        \hline
    \rule[0.4ex]{0pt}{8pt}
    1 & 0 & 0 & 0 \\
        \hline
    \rule[0.4ex]{0pt}{8pt}
    0 & 1 & 0 & 0 \\
  \end{array}
\right),
\eeqs
which is equivalent to
\beqs
\left(
  \begin{array}{cc|cc}
    0 & 0 & 1 & 0 \\
        \hline
    \rule[0.4ex]{0pt}{8pt}
    1 & 0 & 0 & 0 \\
        \hline
    \rule[0.4ex]{0pt}{8pt}
    0 & 1 & 0 & 0 \\
  \end{array}
\right).
\eeqs
 After swapping $\text{A}_j$ and $\text{Q}_j$, the overall effect is a Hadamard gate on $\text{Q}_j$ with $\text{A}_j$ in $|+\>$ up to a Pauli correction depending on the measurement outcomes. (Again, this swap just needs to be recorded in software.)

 Since $\{X_{A_j}Z_{Q_j}\ | \ 1\leq j \leq n\}$ and $\{X_{Q_j}\ | \ 1\leq j \leq n\}$ are both commuting sets, we need just two steps of Pauli measurements and one $4n$-qubit ancilla state for an -H- stage. If Hadamard gates are applied to a set $\mathscr{M}$ of qubits, the required ancilla state  is
\beq\label{eq:had_ancilla}
|\Omega_{\text{H}_\mathscr{M}}\>=\frac{1}{\sqrt{2^{|\mathscr{M}|}}}\prod_{j\in \mathscr{M}}
\left(I_{2N}+X_j\otimes Z_{j+n}\right)|0\>^{\otimes 2n}\otimes |+\>^{\otimes 2n}.
\eeq
It is easy to recognize that it is the same state one obtains after projecting $|0\>^{\otimes 2n}\otimes |+\>^{\otimes 2n}$ to the joint $+1$ eigenspace of $\{X_j\otimes Z_{j+n}\ |\ j\in \mathscr{M}\}$. Thus, it is a CSS state (up to Hadamard gates) stabilized by $\{X_j\otimes Z_{j+n}, Z_j\otimes X_{j+n}\ |\ j\in \mathscr{M}\}$.

%\begin{example}
%
%\end{example}

\subsection{The -C- stage}
The set of measurement operators for a -C- stage is more complicated to find. We first introduce the following lemma that will be used later.
\begin{mylemma}\label{lemma:stabilizer_transform}
Let $L_1$ be an $n\times n$ lower triangular matrix with the diagonal elements being zeros. %For the following matrix
Suppose
\beqs
L=(I_n  \ L_1).
\eeqs
Then there exists  a full-rank matrix $L'=(L_2  \ L_3)$, where $L_2$ and $L_3$ are two $n\times n$ lower triangular matrices,
such that the rows of $L'$ are linear combinations of rows of $L$ and
\beq\label{eq:orthogonal}
 L'
 \left(
   \begin{array}{c}
     I_n \\
     I_n \\
   \end{array}
 \right) = L_2+L_3= I_n.
 \eeq
\end{mylemma}
\begin{proof}
Let $l'_j$ denote the $j$th row vector of $L'$ and $c_p$ be the $p$th column vector of $(I_n \ I_n)^t$.
Equation~(\ref{eq:orthogonal}) is equivalent to
\beq
l'_j c_p = \delta_{jp}, \quad \quad  1\leq j, p\leq n, \label{eq:orthogonal2}
\eeq
where  $\delta$ is the Kroneker delta function.

Let $l_j$ denote the $j$th row vector of $L$. Obviously, $l_1=(1,0,\dots, 0)$, satisfying
$l_1 c_p = \delta_{1p}$. Let $l'_1=l_1$.

It is easy to see that $l_j c_p = 0$ for $p>j$, since $L_1$ is a lower triangular matrix. With all the diagonal elements of $L_1$ being~0, one has
\beq
l_jc_j=1.\label{eq:lc}
\eeq
Define the set $\mathscr{I}_j=\{p\ |\ l_{j}c_p =1, p < j \}$.
%Starting from the second row, we consider the following procedure to construct $L'$:
%for row vector $l_j$ and $j > 1$, consider the value of $l_j c_p$.
For $j=2,\dots,n$, let
\beq
l_{j}'=l_{j} + \sum_{p\in \mathscr{I}_j} l_p'. \label{eq:lp}
\eeq
We also define a matrix $L'^{(j)}$ that contains the rows $l_1',\dots,l_j'$: %during any step $j$ in the construction of $L'$ as
\beqs
L'^{(j)}=\left(
           \begin{array}{c}
             l_1' \\[2pt]
             \vdots \\[2pt]
             l_j' \\
           \end{array}
         \right).
\eeqs
Since $L_1$ is lower triangular, and the summation of  $l_p$  in Eq.~(\ref{eq:lp}) only counts the terms with  $p<j$, $L'^{(j)}$ can be written as
\beqs
L'^{(j)}=\left(L_2^{(j)}\ L_3^{(j)}\right),
\eeqs
where $L_2^{(j)}$ and $L_3^{(j)}$ are also lower triangular matrices. Eventually, we have $L_2=L_2^{(n)}$ and $L_3=L_3^{(n)}$.

It remains to prove Eq.~(\ref{eq:orthogonal2}).
%Now we show this procedure can guarantee $\forall p, l'_j c_p = \delta_{jp}$ via
We prove this by induction.
For $j = 2$, if $l_2 c_1 = 1$, one has $l_2'=l_2 + l_1$. Thus $l_2' c_1 = 0$ and $l_2'c_2=1$, since $l_1c_1=1$ and $l_1c_2=0$. Also, $l_2' c_p=0$ for $p>2$ since $L_2'^{(2)}$ and $L_3'^{(2)}$ are lower triangular matrices. So $l'_2 c_p = \delta_{2p}$ holds for $1\leq p \leq n$.

Now assume $l'_{1}c_p=\delta_{1p}$, $\dots,$ $l'_{j}c_p=\delta_{jp}$ holds. %For $j+1$, one has
Then
\beqs
l_{j+1}'c_q=l_{j+1}c_q + \sum_{p\in \mathscr{I}_{j+1} } l_p'c_q.
\eeqs
Consider $q < j+1$ first. If $l_{j+1}c_q=1$, then $q\in \mathscr{I}_{j+1}$ and
\beqs
\sum_{p\in \mathscr{I}_{j+1} } l_p'c_q = \sum_{p\in \mathscr{I}_{j+1}} \delta_{pq}=1.
\eeqs
Then $l_{j+1}'c_q=0$.
If $l_{j+1}c_q=0$, then $q\notin \mathscr{I}_{j+1}$ and
$\sum_{p\in \mathscr{I}_{j+1}} l_p'c_q = 0$. Again, $l_{j+1}'c_q=0$. When $q=j+1$,
$l'_{j+1}c_{j+1} = l_{j+1}c_{j+1} = 1$ by Eq.~(\ref{eq:lc}).
%due the fact that $L_1$ is lower triangular.
For $q>j+1$, since $L^{(j+1)}_2$ and $L^{(j+1)}_3$ are both lower triangular, $l_{j+1}'c_q = 0$. Thus, $l'_jc_p=\delta_{jp}$ holds for $1\leq j,p\leq n$.

%$(I_n|L')$ can be transformed to $(L_1|L_2)$ such that
%$L_1+L_2=I_s$
%(A constructive proof from first row.)
%Now one can easily see that
%\beq
%(L_1|L_2) J \left(
%             \begin{array}{c}
%               I_n \\
%               \hline
%               I_n \\
%             \end{array}
%           \right) = I_n
%\eeq
\end{proof}

Now we are ready to show that any -C- stage containing $\text{C}(j,l)$ on $\text{Q}_1,\dots , \text{Q}_n$ with $j<l$ can be implemented by a constant number of Pauli measurements.
Unlike the case of -P- stage or -H- stage, we
start from the GSF of an arbitrary -C- circuit with $\text{A}_1,\dots,\text{A}_n$ in $|+\>^{\otimes n}$ state:
\beq
\left(
  \begin{array}{cc|cc}
    {\bf 0}_n  & U & {\bf 0}_n  & {\bf 0}_n \\
    \hline
    \rule[0.4ex]{0pt}{8pt}
    {\bf 0}_n  & 0 & {\bf 0}_n & (U^t)^{-1} \\
    \hline
    \rule[0.4ex]{0pt}{8pt}
    I_n & {\bf 0}_n & {\bf 0}_n & {\bf 0}_n \\
  \end{array}
\right),
\eeq
and try to reduce it to the idle circuit. Meanwhile, we will provide the reverse operations that will effectively implement the target CNOT circuit.

As mentioned before, $U$ is some invertible upper triangular matrix. The GSF is then equivalent to
\beq\label{eq:first_measure_before}
\left(
  \begin{array}{cc|cc}
    U+I_n  & U & {\bf 0}_n  & {\bf 0}_n \\
    \hline
    \rule[0.4ex]{0pt}{8pt}
    {\bf 0}_n  & {\bf 0}_n & {\bf 0}_n & (U^t)^{-1} \\
    \hline
    \rule[0.4ex]{0pt}{8pt}
    I_n & {\bf 0}_n & {\bf 0}_n & {\bf 0}_n \\
  \end{array}
\right)
\eeq
since all the nonzero row vectors of $\left(U+I_n \ \ {\bf 0}_n \ | \ {\bf 0}_n  \ \ {\bf 0}_n\right)$ can be generated by $\left(I_n \ \ {\bf 0}_n \ | \ {\bf 0}_n  \ \ {\bf 0}_n\right)$
and we then add these vectors to the first row.

Since $U$ is of full rank, the diagonal elements of $U+I_n$ must be all zeros.
Observe that  $\left({\bf 0}_n \ {\bf 0}_n \ | \ {I}_n \ \ (U^t)^{-1}+I_n\right)$ commutes with the logical operators and is a symplectic partner of the stabilizer generators.
This can be checked by verifying that
\beqs
\left( {I}_n \ \ \ \ (U^t)^{-1}+I_n\right) \left(U+I_n \ \ \ \ U\right)^t = {\bf 0}_n,
%=& U^{-1} + (U^t)^{-1}U^t + U
\eeqs
and
\beqs
\left(I_n \ {\bf 0}_n \ | \ {\bf 0}_n \ \ {\bf 0}_n\right)J_{2n}\left({\bf 0}_n \ {\bf 0}_n \ | \ {I}_n \ \ (U^t)^{-1}+I_n\right)^t=I_{2n}.
\eeqs
According to Lemma~\ref{lemma:multi_measurement}, one can measure $n$ commuting Pauli operators $\left({\bf 0}_n \ {\bf 0}_n \ | \ {I}_n \ \ (U^t)^{-1}+I_n\right)$ simultaneously. The GSF will then be transformed into
\beq\label{eq:first_measure_after}
\left(
  \begin{array}{cc|cc}
    U+I_n  & U & {\bf 0}_n  & {\bf 0}_n \\
    \hline
    \rule[0.4ex]{0pt}{8pt}
    {\bf 0}_n  & {\bf 0}_n & {\bf 0}_n & (U^t)^{-1} \\
    \hline
    \rule[0.4ex]{0pt}{8pt}
    {\bf 0}_n & {\bf 0}_n & {I}_n & (U^t)^{-1}+I_n \\
  \end{array}
\right).
\eeq
(Meanwhile, we can perform the Pauli measurements $\left(I_n \ \ {\bf 0}_n \ | \ {\bf 0}_n \ \ {\bf 0}_n\right)$ to reverse the process (from Eq.~(\ref{eq:first_measure_after}) to Eq.~(\ref{eq:first_measure_before})).)

Now, adding the third row of Eq.~(\ref{eq:first_measure_after}) to the second row, one can obtain an equivalent GSF
\beq
\left(
  \begin{array}{cc|cc}
    U+I_n  & U & {\bf 0}_n  & {\bf 0}_n \\
    \hline
    \rule[0.4ex]{0pt}{8pt}
    {\bf 0}_n  & {\bf 0}_n & I_n & I_n \\
    \hline
    \rule[0.4ex]{0pt}{8pt}
    {\bf 0}_n & {\bf 0}_n & {I}_n & (U^t)^{-1}+I_n \\
  \end{array}
\right).
\eeq
Let $L=\left( {I}_n \ \ \ L_1\right)$, where $L_1=(U^t)^{-1}+I_n$ is a lower triangular matrix with all the diagonal elements being~0. By Lemma~\ref{lemma:stabilizer_transform}, the GSF can be equivalently transformed into
\beq\label{eq:second_measure_before}
\left(
  \begin{array}{cc|cc}
    U+I_n  & U & {\bf 0}_n  & {\bf 0}_n \\
    \hline
    \rule[0.4ex]{0pt}{8pt}
    {\bf 0}_n  & {\bf 0}_n & I_n & I_n \\
    \hline
    \rule[0.4ex]{0pt}{8pt}
    {\bf 0}_n & {\bf 0}_n & L_2 & L_3 \\
  \end{array}
\right),
\eeq
where $(L_2 \ \ L_3) (I_n \ \ I_n)^t=I_n$. By Lemma~\ref{lemma:multi_measurement} again, one can measure a set of $n$ Pauli operators $(I_n \ \ I_n | \ {\bf 0}_n \ \  {\bf 0}_n)$ simultaneously and transform the GSF into
\beq\label{eq:second_measure_after}
\left(
  \begin{array}{cc|cc}
    U+I_n  & U & {\bf 0}_n  & {\bf 0}_n \\
    \hline
    \rule[0.4ex]{0pt}{8pt}
    {\bf 0}_n  & {\bf 0}_n & I_n & I_n \\
    \hline
    \rule[0.4ex]{0pt}{8pt}
    I_n & I_n & {\bf 0}_n & {\bf 0}_n \\
  \end{array}
\right).
\eeq
Meanwhile, measuring $\left({\bf 0}_n \ \ {\bf 0}_n \ | \ L_2 \ \ L_3\right)$  will transfer the GSF of Eq.~(\ref{eq:second_measure_after}) into Eq.~(\ref{eq:second_measure_before}). Note that the measurement of $\left({\bf 0}_n \ \ {\bf 0}_n \ | \ L_2 \ \ L_3\right)$ is equivalent to measuring $\left({\bf 0}_n \ \ {\bf 0}_n \ | \ {I}_n \ \ (U^t)^{-1}+I_n \right)$.

Now, since the stabilizer generators in Eq.~(\ref{eq:second_measure_after})  are of the form $\left(I_n \ \ I_n \ | \ {\bf 0}_n \ \ {\bf 0}_n\right)$, one can add $\left(U+I_n \ \ U+I_n \ | \ {\bf 0}_n  \ \ {\bf 0}_n\right)$ to the first row of Eq.~(\ref{eq:second_measure_after}), which equivalently reduces the GSF to:
\beq\label{eq:third_measure_before}
\left(
  \begin{array}{cc|cc}
    {\bf 0}_n  & I_n & {\bf 0}_n  & {\bf 0}_n \\
    \hline
    \rule[0.4ex]{0pt}{8pt}
    {\bf 0}_n  & {\bf 0}_n & I_n & I_n \\
    \hline
    \rule[0.4ex]{0pt}{8pt}
    I_n & I_n & {\bf 0}_n & {\bf 0}_n \\
  \end{array}
\right).
\eeq

The final step is to eliminate the left-most $I_n$ in the second row of Eq.~(\ref{eq:third_measure_before}). This can be done by measuring  $\left({\bf 0}_n \ \ {\bf 0}_n \ | \ I_n  \ \ {\bf 0}_n\right)$ and adding the third row to the second. This will then transform the GSF into the second matrix in Eq.~(\ref{eq:identity}):
\beq\label{eq:third_measure_after}
\left(
  \begin{array}{cc|cc}
    {\bf 0}_n  & I_n & {\bf 0}_n  & {\bf 0}_n \\
    \hline
    \rule[0.4ex]{0pt}{8pt}
    {\bf 0}_n  & {\bf 0}_n & {\bf 0}_n & I_n \\
    \hline
    \rule[0.4ex]{0pt}{8pt}
    {\bf 0}_n & {\bf 0}_n & I_n & {\bf 0}_n \\
  \end{array}
\right),
\eeq
Meanwhile, one can measure the set of $n$ Pauli operators $\left(I_n \ \ I_n \ | \ {\bf 0}_n \ \ {\bf 0}_n\right)$ to transform Eq.~(\ref{eq:third_measure_after}) to Eq.~(\ref{eq:third_measure_before}).

To reverse the whole procedure above and start from Eq.~(\ref{eq:third_measure_after}), we initially set $\text{A}_1,\dots, \text{A}_n$ to $|0\>^{\otimes n}$ and perform the following three sets of Pauli measurements simultaneously:
\beq\label{eq:three_measurements}
\begin{split}
&1.~\left(I_n \ \ I_n \ | \ {\bf 0}_n \ \ {\bf 0}_n\right),\\
&2.~\left({\bf 0}_n \ \ {\bf 0}_n \ | \ {I}_n \ \ (U^t)^{-1}+I_n \right),\\
&3.~\left(I_n \ \ {\bf 0}_n \ | \ {\bf 0}_n \ \ {\bf 0}_n\right).
\end{split}
\eeq
The non-trivial measurements 1 and 2 need two $2n$-qubit CSS ancilla states (see Eq.~(\ref{eq:state_multi_measurement})), which are
\beq\label{eq:Cnot_ancilla_1}
|\Omega_{\text{C}_1}\>=\frac{1}{\sqrt{2^n}}\prod_{j=1}^n \left(I+X_j X_{j+n}\right)|0\>^{\otimes 2n}
\eeq
and
\beq\label{eq:Cnot_ancilla_2}
|\Omega_{\text{C}_2}\>=\frac{1}{\sqrt{2^n}}\prod_{j=1}^n \left(I+Z^{{\bf u}_j}\right)|+\>^{\otimes 2n},
\eeq
where ${\bf u}_j$ is the $j$th row of   $\left({I}_n \ \ (U^t)^{-1}+I_n \right)$.
%\beq
%|\Omega_{{\text{C}}\>=\frac{1}{\sqrt{2^n}}\prod_{j}
%\left(I_{2N}+iX_j\otimes Z_j\right)|0\>^{\otimes N}\otimes |+\>^{\otimes N}.
%\eeq
The first ancilla is actually an $n$-fold tensor product of Bell states. The second ancilla is the key resource state in our procedure to reduce the depth of -C- stage computation. Both of them are $2n$-qubit CSS states.
The third step is a trivial bitwise measurement on $\text{A}_1,\dots,\text{A}_n$ in the $X$ basis and can be directly done without additional ancillas.
The net effect is the desired \text{-C-} stage computation acting on $\text{Q}_{1}, \dots, \text{Q}_n$, and the auxiliary qubits $\text{A}_{1}, \dots, \text{A}_n$ are reset to $|+\>^{\otimes n}$ (up to $Z$ corrections). One can transfer $\text{A}_{1}, \dots, \text{A}_n$ back into $|0\>^{\otimes n}$ or just keep them and start with $|+\>^{\otimes n}$ for the next stage. The procedure with $\text{A}_{1}, \dots, \text{A}_n$ initially in $|+\>^{\otimes n}$ state for -C- stage is similar. As a conclusion, one has the following theorem:
%\beq
%\left(
%  \begin{array}{cc|cc}
%    {\bf 0}_n  & I_n & {\bf 0}_n  & {\bf 0}_n \\
%    \hline
%    \rule[0.4ex]{0pt}{8pt}
%    {\bf 0}_n  & {\bf 0}_n & {\bf 0}_n & I_n \\
%    \hline
%    \rule[0.4ex]{0pt}{8pt}
%    I_n & {\bf 0}_n & {\bf 0}_n & {\bf 0}_n \\
%  \end{array}
%\right),
%\eeq
\begin{mytheorem}\label{thm:cnot_measurement}
For a set of $2n$ qubits $\text{A}_1,\dots, \text{A}_n, \text{Q}_1,\dots, \text{Q}_n$, where $\text{A}_1,\dots, \text{A}_n$ are initially in state $|0\>^{\otimes n}$ or $|+\>^{\otimes n}$, any -C- stage circuit containing only $\text{C}(j,l)$ with $j<l$ on $\text{Q}_1,\dots, \text{Q}_n$ can be realized via three steps of Pauli measurements on these $2n$ qubits using two $2n$-qubit CSS states.
\end{mytheorem}

Note that the procedure to construct -C- stages via Pauli measurements in Theorem~\ref{thm:cnot_measurement} also works with  additional permutations
on qubits $\{\text{Q}_1,\dots, \text{Q}_n\}$.
%\begin{mycorollary}
%Any -C- stage which can reduce to the stage containing only $\text{C}(j,l)$ with $j<l$ on $Q_1,\dots, Q_n$ through qubits permutation can be obtained via three steps of Pauli measurements using two $2n$-qubit two-colorable graph states.
%\end{mycorollary}
Thus, -C- stages with symplectic matrices of the form $\pi M_C \pi^{-1}$ can also be computed using only three steps of measurements.

To implement an arbitrary Clifford circuit in the form of Eq.~(\ref{eq:used_bruhat}), one needs four -P- stages, four -C-  stages, one -H- stage, and permutations of qubits (which can be done in software by keeping  records). Overall, it requires five $4n$-qubit ancilla states of the form (\ref{eq:state_multi_measurement}) and eight $2n$-qubit CSS states. Crucially, the ancilla states are all equivalent to CSS states up to single-qubit Clifford gate operations. This gives us the main result of the paper:
\begin{mytheorem}\label{thm:main}
Any Clifford circuit on $n$ qubits can be implemented up to a qubit permutation by 22 steps of Pauli measurements, by introducing $n$ auxiliary qubits and preparing five $4n$-qubit stabilizer states and eight $2n$-qubit stabilizer state.
\end{mytheorem}

Like gate teleporation~\cite{Gottesman:1999:390}, this theorem implies that the gate complexity of the circuit is now completely dominated by the preparation of these CSS states. These resource states can be prepared using a stabilizer circuit of $O(n^2)$ gates with depth $O(n)$. On the other hand, note that these CSS states are equivalent to two-colorable graph states up to local Clifford operations. Thus, they can be approximated as non-degenerate ground states of two-body Hamiltonians~\cite{van2008graph,darmawan2014graph}. This fact may help to prepare these CSS states in an adiabatic manner.

Note that the total number of qubits used increases by a factor of five. The protocol can save the real computation time of stabilizer circuits by off-line preparation of CSS states.

\section{Fault-tolerant Ancilla States preparation}
Quantum states are vulnerable to noise, which is the main obstacle to building large-scale quantum computers. The solution requires encoding the quantum state into some quantum error correcting code and performing fault-tolerant quantum computation; see~\cite{QECbook:2013} for details.

The Pauli measurement circuit in Fig.~\ref{fig:measurement} is naturally compatible with Steane syndrome extraction when $2n$ qubits $\text{A}_1,\dots,\text{A}_n,\text{Q}_{1},\dots \text{Q}_n$ are encoded in an $[[n', k\geq2n, d]]$ CSS code $\mathcal{Q}$, while $4n$ additional qubits are also encoded into two blocks of the same code $\mathcal{Q}$. That is to say, the syndrome measurements and Pauli measurements can be done simultaneously through the circuit in Fig.~\ref{fig:measurement} at the logical level.

The ancilla states required in this paper, (\ref{eq:phase_ancilla}),~(\ref{eq:had_ancilla}),~(\ref{eq:Cnot_ancilla_1}) and~(\ref{eq:Cnot_ancilla_2}), are all CSS states up to single-qubit Clifford operations. Fortunately, when encoded in $\mathcal{Q}$, these states can all be distilled fault-tolerantly using a structure based on classical error-correcting codes, with a constant overhead for the purpose of FTQC~\cite{zheng2017efficient}. In such a scenario, one can fault-tolerantly compute Clifford circuits in $O(1)$ steps.
In addition, the complexity of state preparation here is at the \emph{physical} level rather than the \emph{logical} level. Thus our method of Clifford circuit computation can have a speedup up by a factor of up to $O(d)$ compared to implementing the circuits directly on the data block in a code deformation manner.

%\section{Graph states as resource state for depth reduction}
%In this section, we show that the ancilla states to do Pauli measurements for
%\begin{thm}
%The ancilla stabilizer states used implement -C- stages are graph states up to a correction of Hadamard gates applying on subset of qubits.
%\end{thm}
%\begin{proof}
%One needs to measure Pauli operators three times as in Eq.~(\ref{eq:three_measurements}). The first measurements requires a set $n$ of Bell states, which can be generated by graph states easy by applying Hadamard gates. The third measurement can be done directly on $\text{Q}_1,\dots, \text{Q}_n$. Here, we show that the ancilla states required for second measurements can also be generated by graph states according to a Hadamard correction on some qubits.
%\end{proof}

\section{Discussion and Conclusions}\label{sec:discussion}

%However, all Clifford gates can be fault-tolerantly implemented via measurements of Pauli operators with the help of special ancillas and possibly followed by measurements of Pauli operators at the end~\cite{brun2015teleportation}; meanwhile, all ancilla states needed for such measurements are stabilizer states, which can be prepared efficiently and fault-tolerantly with some constant factor of overhead~\cite{Ancilla_distillation_1,zheng2017efficient}.

In this paper, we proposed a method to compute Clifford circuits by  a constant number of steps of Pauli measurement, which is suitable for Steane syndrome measurements. Consequently, the depth of the circuit is reduced to $O(1)$. It requires $n$ auxiliary qubits and preparation of five $4n$-qubit and eight $2n$-qubit CSS states.

%(Discuss the complexity transfer and why it is interesting, including the easy to prepare ancilla stabilizer state, easy for FTQC, only small set of ancilla needed to prepare for FTQC)
The gate complexity is then completely transferred into the off-line preparation of these CSS states. The overall gate complexity, including the ancilla preparation, is $O(n^2)$ with a depth of $O(n)$. It seems initially at least as difficult as   the direct implementation. However, preparing these known states is much simpler than implementing gate operations on unknown states. At the physical level, these states can be prepared via a gap-protected adiabatical process. When encoding in CSS codes, they can be prepared by distillation and postselection, with a constant overhead in practice, which is especially interesting in the scenario of FTQC using multi-qubit CSS codes~\cite{brun2015teleportation}.
%A direct implication is that the stabilizer state of the form (\ref{eq:state_multi_measurement}) can be regarded as resource to speed up quantum computation. One can imagine in the future market of quantum computation, vendors can prepare the these resource and sell to users to save their computation time.
There, in general, one cannot find a fault-tolerant way to directly implement gates on logical qubits for a given multi-qubit CSS codes.
%Gate teleportation is also difficult since encoded states of the form %Eq.~(\ref{eq:gate_teleportation}) are hard to prepare fault-tolerantly.
By contrast, measurements of logical Pauli operators are very easy for any CSS codes by fault-tolerantly preparing logical CSS states (up to single-qubit Clifford operation) with only constant resource overhead~\cite{Ancilla_distillation_1,zheng2017efficient}.

The method of this paper implies that the number of different encoded ancilla states can be reduced from $O(n^2/\log n)$ to $O(1)$ for a given Clifford circuit, which can greatly simplify the fault-tolerant preparation procedure through high throughput distillation, if the number of Clifford circuits required is limited. This makes FTQC based on high rate multi-qubit error correcting codes more promising.

%As an example, for protocol in Ref.~\cite{brun2015teleportation}, one need to prepare $O(n^2/\log n)$ stabilizer states previously for a stabilizer circuit.

%DISCUSSION
%(especially for the application in large scale quantum simulation and FTQC)
%Like magic state, the stabilizer ancilla states in this paper can be regarded as resource states --- it can greatly speed up the real time quantum computation.
%
%(relatively easy when circuit is noisy)
%Prepare off-line high quality complicated state is relatively easy. once this state is prepared,

\acknowledgments
The funding support from the National Research Foundation \& Ministry of Education, Singapore, is acknowledged. This work is also supported by the National Research Foundation of Singapore and Yale-NUS College (through grant number IG14-LR001 and a startup grant). TAB was supported by NSF Grants No. CCF-1421078 and No. MPS-1719778, and by an IBM Einstein Fellowship at the Institute for Advanced Study.

\newpage
\bibliographystyle{apsrev4-1}
%\bibliography{C:/Users/Ò»´Ï/Dropbox/bib/refs}
%\bibliography{refs}
%

\end{document}